\documentclass[a4paper,twocolumn,english,showpacs,nofootinbib,nobibnotes,aps,pra,floatfix,superscriptaddress]{revtex4-1}
\usepackage[utf8]{inputenc}
\usepackage{amsmath}
\usepackage{amsthm}
\usepackage{amssymb}
\usepackage{graphicx}
\usepackage[caption=false]{subfig}
\usepackage{xparse}
\usepackage{xcolor}

\makeatletter

\newcounter{thmc}

\newtheorem{proposition}[thmc]{Proposition}

\usepackage{braket}
\usepackage{dsfont}

\makeatother

\usepackage{babel}

\global\long\def\trace{\operatorname{Tr}}
\global\long\def\ketbra#1{\ket{#1}\!\bra{#1}}
\global\long\def\ketbraa#1#2{\ket{#1}\!\bra{#2}}

\global\long\def\one{\mathds{1}}

\newcommand{\kommentar}[1]{}

\newcommand{\SR}{\operatorname{SR}}
\newcommand{\SN}{\operatorname{SN}}

\NewDocumentCommand\opti{smmm>{\SplitList{;}}m} {
\begingroup%
\setlength{\belowdisplayskip}{-0.6\baselineskip}%
\IfBooleanTF{#1}{%
    \begin{alignat*}{2}
        & \underset{#3}{\text{#2}} & & #4 \\
        & \text{subject to~~}
        \ProcessList{#5}{ \insertopticonst }
        & &
    \end{alignat*}%
    }{%
    \begin{alignat}{2}
        & \underset{#3}{\text{#2}} & & #4 \\
        & \text{subject to~~}
        \ProcessList{#5}{ \insertopticonst }
        & & \nonumber
    \end{alignat}%
    }%
\endgroup%
}%
\newcommand\insertopticonst[1]{& & #1\\&}

\begin{document}

\title{Construction of efficient Schmidt number
witnesses for\\ high-dimensional quantum states}

\author{Nikolai Wyderka}
\affiliation{Institut für Theoretische Physik III, Heinrich-Heine-Universität Düsseldorf, Universitätsstr. 1, D-40225 Düsseldorf, Germany}

\author{Giovanni Chesi}
\affiliation{Istituto Nazionale di Fisica Nucleare Sezione di Pavia, Via Agostino Bassi 6, I-27100 Pavia, Italy}

\author{Hermann Kampermann}
\affiliation{Institut für Theoretische Physik III, Heinrich-Heine-Universität Düsseldorf, Universitätsstr. 1, D-40225 Düsseldorf, Germany}

\author{Chiara Macchiavello}
\affiliation{Istituto Nazionale di Fisica Nucleare Sezione di Pavia, Via Agostino Bassi 6, I-27100 Pavia, Italy}

\author{Dagmar Bruß}
\affiliation{Institut für Theoretische Physik III, Heinrich-Heine-Universität Düsseldorf, Universitätsstr. 1, D-40225 Düsseldorf, Germany}

\date{\today}

\begin{abstract}
Recent progress in quantum optics has led to setups that are able to prepare high-dimensional quantum states for quantum information processing tasks. As such, it is of importance to benchmark the states generated by these setups in terms of their quantum mechanical properties, such as their Schmidt numbers, i.e., the number of entangled degrees of freedom.

In this paper, we develop an iterative algorithm that finds Schmidt number witnesses tailored to the measurements available in specific experimental setups. We then apply the algorithm to find a witness that requires the measurement of a number of density matrix elements that scales linearly with the local dimension of the system. As a concrete example, we apply our construction method
to an implementation with photonic temporal modes.
\end{abstract}
\maketitle

\section{Introduction}

Over the course of the last decades, different physical platforms have been developed to perform quantum informational tasks. The requirements for such platforms have been formalized at various occasions. One of the most often used list of requirements was formulated by DiVincenzo, and includes among other things the capability of robust state preparation, state manipulation and measurements \cite{divincenzo2000physical}. While most platforms available today are limited to the generation and manipulation of qubits, it has become more and more evident over the last years that higher-dimensional systems can provide advantages in specific tasks including quantum communication \cite{bruss2002optimal, cerf2002security, durt2004security, sheridan2010security, coles2016numerical,ferenczi2012symmetries} and quantum computing (see Ref.~\cite{wang2020qudits} and references therein). However, embedding these additional degrees of freedom in trapped ions or superconducting qubits has been challenging \cite{low2020practical, strauch2011quantum}. In contrast, time-energy degrees of freedom in photonic states can be used to encode very high-dimensional quantum states. Furthermore, using quantum pulse gates allows for robust and precise manipulation and measurement of these signals, and therefore provides a viable candidate for a platform for high-dimensional quantum tasks \cite{brecht2015photon}.

In this paper, we develop theoretical tools that help to characterize the state preparation capabilities of these and similar setups. Recent work focused on the construction of high-dimensional entanglement witnesses for these setups that allow for efficient entanglement detection of the prepared quantum states \cite{sciara2019universal, riccardi2020optimal}. In this work, we construct easily measurable Schmidt number witnesses for such setups. Apart from detecting whether a state is entangled, the Schmidt number of a state carries additional information about the dimensionality of the entanglement. We develop an algorithm to generate witness candidates that use only few of the experimentally available measurement settings.
Consequently, we then apply the algorithm to the measurements available in the setup described in Ref.~\cite{brecht2015photon} and show that the obtained observable indeed is a proper Schmidt number witness, i.e., it certifies the Schmidt number of the generated quantum state. While proposals exist to certify Schmidt numbers with measurements in only two local bases \cite{bavaresco2018measurements}, these ideas do not apply to the setup at hand, as here only certain linear combinations of the entries of the density matrix can be measured. 

The paper is organized as follows. In section~\ref{sec:SNW} we review the notion of Schmidt numbers and Schmidt number witnesses, in section~\ref{sec:algorithm} we develop and explain the algorithm that generates Schmidt number witness candidates. In Section~\ref{sec:setup} we apply the algorithm to obtain a Schmidt number witness using only $\mathcal{O}(d)$ measurement settings to certify the Schmidt number of the generated states, and compare its noise robustness to that of other Schmidt number witnesses. In the Appendix, we provide an explicit construction
for an experiment using photonic temporal modes.

\section{Schmidt numbers and Schmidt number witnesses}\label{sec:SNW}

Throughout this paper, we consider bipartite quantum systems in $\mathbb{C}^d \otimes \mathbb{C}^d$. We denote the set of linear maps from $\mathbb{C}^d$ to itself by $\mathcal{M}_d$.

The Schmidt number of a bipartite quantum state is an entanglement measure that is related to the hardness of generating a quantum state using local operations and classical communication \cite{nielsen1999conditions}. For pure states, it is defined as the number of non-vanishing Schmidt coefficients in the Schmidt decomposition of the state, i.e., for every bipartite quantum state $\ket{\psi}$, written in the computational basis as
\begin{align}
    \ket{\psi} = \sum_{i,j=0}^{d-1} c_{ij} \ket{i}_A\ket{j}_B,
\end{align}
one can find local orthonormnal bases, $\{\ket{\underline{i}}_A\}, \{\ket{\underline{j}}_B\}$ for subsystems $A$ and $B$, respectively, such that
\begin{align}
    \ket{\psi} = \sum_{i=0}^{k-1} \lambda_i \ket{\underline{i}}_A\ket{\underline{i}}_B,
\end{align}
where $\lambda_i \in \mathbb{R}, \lambda_i >0$ and $\sum_{i=0}^{k-1} \lambda_i^2 = 1$. The  $k\leq d$ non-vanishing numbers $\lambda_i$ are called Schmidt coefficients of $\ket{\psi}$, and $k$ is called the Schmidt rank of $\ket{\psi}$, or $\SR(\ket{\psi})$ in short \cite{guhne2009entanglement}. 
In order to generalize this measure to mixed states $\rho$, one uses the convex roof construction of the Schmidt rank, called Schmidt number, or $\SN(\rho)$ \cite{terhal2000schmidt}:
\begin{align}
    \SN(\rho) := \min_{\rho = \sum_i p_i \ketbra{\psi_i}} \max_i \SR(\ket{\psi_i}),
\end{align}
i.e., it is given by the maximal Schmidt rank within a given decomposition of $\rho$, minimized over all decompositions.
Due to the minimization, it is usually hard to calculate the Schmidt number of a given quantum state.

For a bipartite system of dimension $d\times d$, the maximal Schmidt number is given by $d$, and one can define the sets of Schmidt number $k$ as
\begin{align}
    S_k = \{\rho\,:\,\SN(\rho) \leq k\}.
\end{align}
Clearly, $S_k \subset S_{k+1}$, and $S_1$ is the usual set of separable states. The maximally entangled state $\ket{\phi^+} = \frac{1}{\sqrt{d}} \sum_{i=0}^{d-1} \ket{ii}$ is member of $S_d$, but not $S_{d-1}$, and it can be shown that \cite{terhal2000schmidt, horodecki1999reduction}
\begin{align}\label{eq:phipoverlap}
    \braket{\phi^+|\rho_k|\phi^+} \leq \frac kd
\end{align}
for all states $\rho_k \in S_k$.

In order to certify a specific Schmidt number experimentally, it is useful to define an analogon to entanglement witnesses for Schmidt numbers \cite{sanpera2001schmidt}. We say that an observable $W_k$ is a Schmidt-number-$(k+1)$ witness, if
\begin{itemize}
    \item $\trace(W_k \rho_k) \geq 0$ for all $\rho_k \in S_k$,
    \item $\trace(W_k \rho) < 0$ for at least one $\rho$.
\end{itemize}
Thus, whenever one finds in an experiment that $\trace(W_k \rho) < 0$, then $\rho$ must have at least Schmidt number $k+1$. Note that for $k=1$, we recover the usual notion of an entanglement witness \cite{guhne2009entanglement}. An important experimental advantage of using theses witnesses lies in the fact that a negative expectation value certifies a certain Schmidt number for any experimental input state, pure or mixed. 
In order to certify that a given observable is a Schmidt-number-$(k+1)$ witness, it is sufficient to minimize its overlap w.r.t.~pure states in $S_k$ and show that the minimum is non-negative. This can be seen from the fact that the sets $S_k$ are compact, thus, one can find a (potentially mixed) optimal quantum state $\rho_k^\star$ such that $c_k := \min_{\rho_k \in S_k} \trace(W_k \rho_k) = \trace(W_k\rho_k^\star)$. As $\SN(\rho_k^\star) \leq k$, it exhibits a decomposition $\rho_k^\star = \sum_i p_i \ketbra{\psi_i}$ with $\SR(\ket{\psi_i}) \leq k$. Thus, $c_k = \sum_i p_i \braket{\psi_i | W_k | \psi_i} \geq \sum_i p_i c_k = c_k$, implying that the inequality is an equality and each of the pure $\ket{\psi_i}$ achieves the same minimal value of $c_k$.

The observation of Eq.~(\ref{eq:phipoverlap}) can be directly transformed into a Schmidt-number-$(k+1)$ witness via the observable  \cite{sanpera2001schmidt} \begin{align}\label{eq:standardwitness}
    W_k = \one_{d^2} - \frac dk \ketbra{\phi^+}.
\end{align}
We refer to this witness as the standard Schmidt-number witness, and we will compare our constructions to this one in the end.

The problem remains that in general, the minimization over pure states with fixed Schmidt rank remains challenging, and in many cases no analytical solution can be found. This can be remedied by relaxing the optimization slightly.

To that end, note that there exists a characterization of the set $S_k$ in terms of positive maps \cite{terhal2000schmidt}: It holds that $\rho \in S_k$ if and only if $(\one_{d} \otimes \Lambda_k)(\rho) \geq 0$ for all $k$-positive maps $\Lambda_k\,:\,\mathcal{M}_d \rightarrow \mathcal{M}_d$, where $\one_d$ denotes the identity map in dimension $d$. A map $\Lambda_k$ is called $k$-positive, if $\one_{k} \otimes \Lambda_k$ is a positive map.

This characterization is useful, as it allows to define slightly larger sets than $S_k$, which can be characterized with less effort. To that end, we define the generalized reduction map \cite{guhne2009entanglement, terhal2000schmidt}
\begin{align}\label{eq:reductionmap}
    R_p(\rho_A) = \trace(\rho_A)\one_d - p\rho_A,
\end{align}
where $\rho_A \in \mathcal{M}_d$ is a single qudit mixed state. It was shown that $R_p$ is $k$-positive, but not $k+1$-positive, iff $p\in(\frac1{k+1}, \frac1k]$ \cite{terhal2000schmidt, tomiyama1985geometry}. For $p=1$ one recovers the usual reduction map. We use it to define
\begin{align}\label{eq:SkR}
    S_k^R = \{\rho\,:\,(\one_d \otimes R_{\frac1k})(\rho) \geq 0\}.
\end{align}
Using the relation between $k$-positive maps and states of Schmidt number $k$, it is clear that $S_k \subset S_k^R$. Thus, we can use these sets as an outer approximation of $S_k$, as they can be characterized by semi-definite constraints. The general embedding situation is displayed in Fig.~\ref{fig:sets1}.

\begin{figure}
    \centering
    \includegraphics[width=1.0\columnwidth]{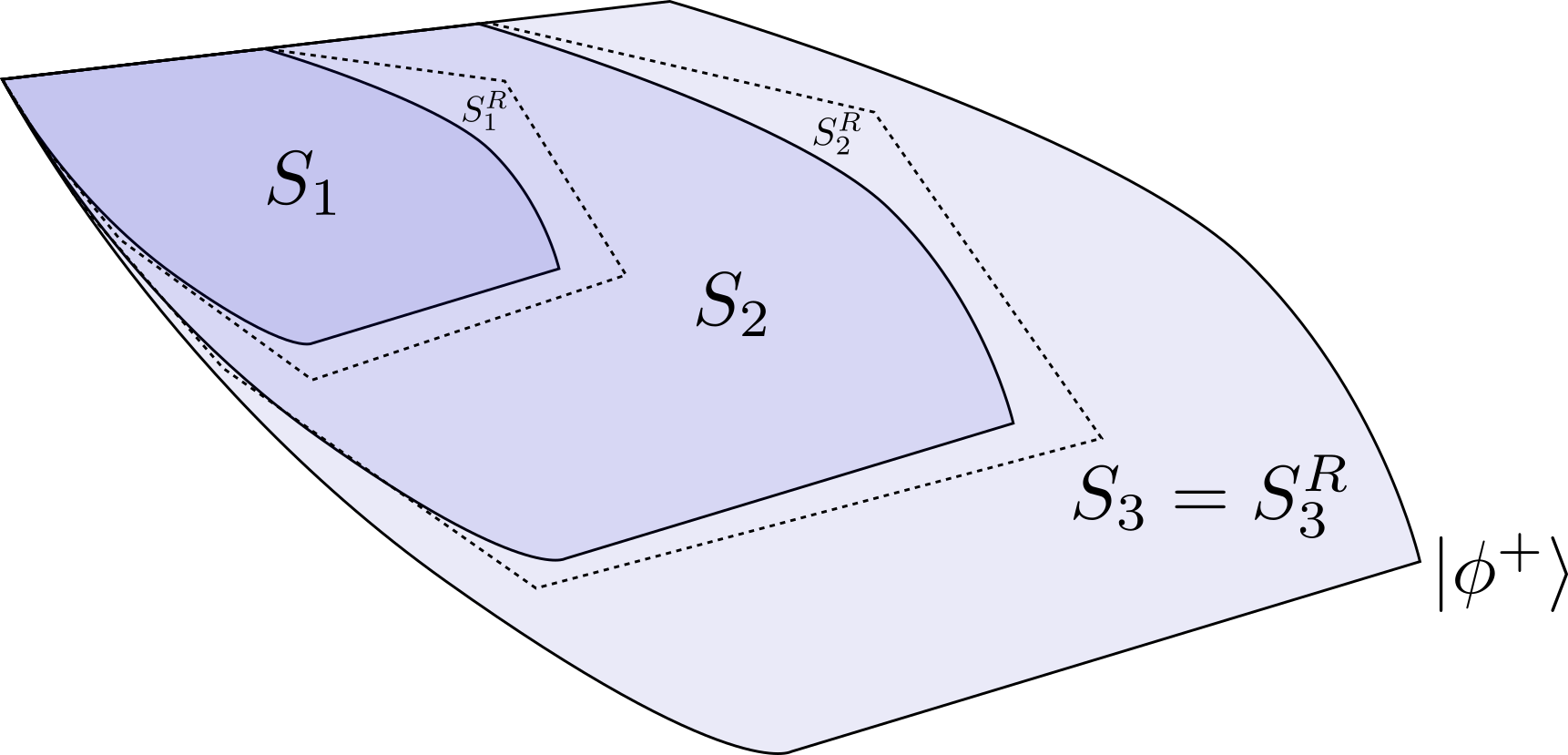}
    \caption{The sets $S_k$ of states of Schmidt number $k$ for $d=3$, as well as their outer approximations $S_k^R$ defined in Eq.~(\ref{eq:SkR}). }
    \label{fig:sets1}
\end{figure}

We use this fact to find a lower bound on the optimal value of $c_k = \min_{\rho_k \in S_k}\trace(W_k \rho_k) \geq \min_{\rho_k \in S_k^R}\trace(W_k \rho_k)$. The latter optimization over $S_k^R$ is in fact a semi-definite program (SDP), that is efficiently solvable on a computer and facilitates further analytical insights by converting it from its primal into its dual form, given by
\opti{max}{y, S\in\mathcal{M}_d}{y\label{eq:dualsdp}}{S^\dagger=S, S\geq0,\nonumber;W_k -(\one \otimes R_{\frac1k})(S)\geq y\one.\nonumber} 
This is equivalent to $\max_{S\geq 0} \lambda_{\text{min}}[W_k - (\one \otimes R_{\frac1k})(S)]$. For this SDP it is possible to show strong duality, meaning that both the primal and the dual optimal values coincide. Strong duality holds thanks to Slater's condition \cite{vandenberghe1996semidefinite}: Both the primal and the dual have full rank feasible points, namely $\rho=\one / d^2$ and $S=\one$. Indeed, any feasible choice of $S$ provides a proper lower bound for $c_k$, allowing one to shift $W_k$ such that it constitutes a proper Schmidt-number-$(k+1)$ witness.

\begin{figure}
    \centering
    \includegraphics[width=0.7\columnwidth]{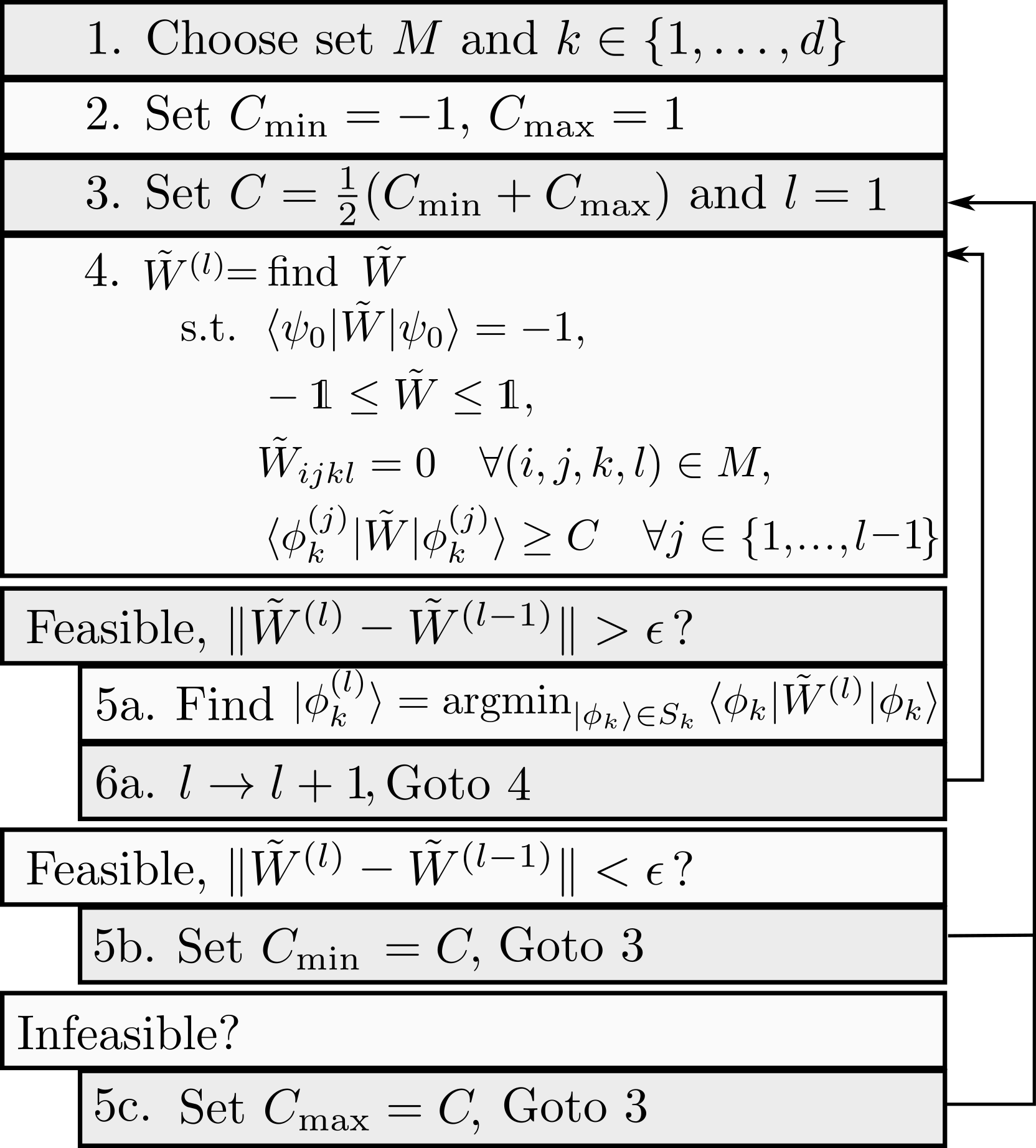}
    \caption{Pictorial representation of the algorithm, including a divide-and-conquer scheme to obtain the optimal choice of $C$. The algorithm stops, when $C_{\min}$ and $C_{\max}$ are close together, and the last feasible $\tilde{W}$ is returned.}
    \label{fig:algo}
\end{figure}

\section{Construction of Schmidt number witnesses}\label{sec:algorithm}

We now turn to the question of how to construct a Schmidt-number-$(k+1)$ witness for a given setup and a given target state $\ket{\psi_0}$ that the setup aims to prepare.\footnote{While we consider here a pure target state $\ket{\psi_0}$, the algorithm works just as well for mixed target states.} The goal is to minimize the experimental requirements to evaluate $\trace(W_k \rho_0)$, i.e., to find a witness that requires less experimental effort to be measured than the standard witness. Thus, we limit the number of required measurement settings. Motivated from the specific setup considered in Appendix~\ref{app:tm}, we restrict ourselves to projective measurements in the standard basis. To that end, we write our witness as $W_k=\sum_{i,j,k,l=0}^{d-1} W_{ijkl} \ketbraa{ij}{kl}$ and denote the set of the indices of those coefficients which are experimentally accessible by $M \subset \{(i,j,k,l)\,:\,0\leq i,j,k,l < d\}$. A suitable choice of $M$ that reflects the experimental capabilities of course depends on a specific experimental setup, and the quality of the obtained witness will depend on $M$.

As a first candidate for our witness, we start by formulating the following semi-definite program:

\opti{find}{}{\tilde{W}\label{eq:sdpstep1}}{\tilde{W}^{\dagger} = \tilde{W},\nonumber;\braket{\psi_0|\tilde{W}|\psi_0} = -1,\nonumber;-\one \leq \tilde{W} \leq \one,\nonumber;\tilde{W}_{ijkl} = 0\quad\forall (i,j,k,l)\notin M.\nonumber}

The result of this program is an operator $\tilde{W}^{(1)}$ that has as a negative eigenvalue of $-1$ corresponding to the state $\ket{\psi_0}$, which is the state we want to detect. The constraint in the last line restricts the operator to contain only terms that our measurement apparatus can measure. Note that we introduce the upper bound of $\tilde{W}^{(1)} \leq \one$ to ensure a bounded result. Of course, $\tilde{W}^{(1)}$ lacks the property of having positive expectation value w.r.t. states of Schmidt number $k$. To implement that, we first fix $k$ and numerically minimize $\braket{\phi_{k}|\tilde{W}^{(1)}|\phi_{k}}$ over pure states $\ket{\phi_{k}}$ of Schmidt rank bounded by $k$. This is achieved by explicitly parameterizing the Schmidt bases and coefficients and using a gradient descent \cite{boyd2003subgradient} or a see-saw algorithm \cite{spall2012cyclic}. The latter fixes all but one of the constituents, optimizing only one of them at a time, and it usually converges within a few cycles (one cycle being an optimization of each of the constituents), keeping the total runtime very short for $k\leq d \leq 8$.  We use this optimized (but possibly not optimal) state $\ket{\phi_{k}^{(1)}}$ to add another constraint to the SDP in Eq.~(\ref{eq:sdpstep1}), namely
\begin{align}\label{eq:constraintC}
    \braket{\phi_{k}^{(1)} | \tilde{W} | \phi_{k}^{(1)}} \geq C,
\end{align}
where $C\in (-1,1]$, as to separate the state $\ket{\phi_{k}^{(1)}}$ from the target state $\ket{\psi_0}$ as much as possible. This threshold value of $C$ is an input to the algorithm and in general cannot be chosen equal to $0$, as the constraint $\tilde{W} \leq \one$, introduced to make the optimization bounded, constrains the spectrum of $\tilde{W}$. Therefore, $C$ should be chosen as large as possible without rendering any of the SDPs in the algorithm infeasible. In practice, we use a divide-and-conquer scheme to obtain an optimal choice of $C$: We begin with a lower bound $C_{\min} = -1$, which is guaranteed to yield a feasible operator $\tilde{W}$, and an upper bound $C_{\max} = 1$, which probably renders the SDP infeasible at some point. We then choose $C=\frac12(C_{\min} + C_{\max}) = 0$, which stays fixed over the course of the algorithm until it stops as described below, after which we update the boundaries accordingly.

Having added the constraint in Eq.~(\ref{eq:constraintC}), we rerun the SDP to obtain a new candidate witness $\tilde{W}^{(2)}$, minimize again numerically the overlap with Schmidt rank $k$-states, yielding the state $\ket{\phi_{k}^{(2)}}$. We add the constraint $\braket{\phi_{k}^{(2)} | \tilde{W} | \phi_{k}^{(2)}} \geq C$ with the same $C$ as before to the SDP, and repeat the whole process, until 
\begin{itemize}
    \item either, the SDP turns infeasible at some point. In this case, the threshold $C$ was chosen too large and has to be reduced by setting $C_{\max} = C$. We then rerun the whole algorithm.
    \item or, the series of SDP converges to some $\tilde{W}^{(\infty)}$, for which no more states $\ket{\phi_{k}}$ with expectation value below $C$ can be found. In this case, $C_{\min}$ can be increased to $C$ and the algorithm rerun again.
\end{itemize}
The algorithm is described compactly in Fig.~\ref{fig:algo} and illustrated graphically in Fig.~\ref{fig:sets2}.
Note that the algorithm usually stops after few iterations, when $C_{\min}$ and $C_{\max}$ are sufficiently close together, providing an optimal choice of $C$.

The result of the algorithm is an operator $\tilde{W}^{(\infty)}$, which probably has the property of $\braket{\phi_k|\tilde{W}^{(\infty)}|\phi_k} \geq C_k$ for all Schmidt rank $k$ states $\ket{\phi_k}$. However, this is not guaranteed yet, as the optimization over these states within the algorithm is just numerical and there is no way to be sure that the true minimum was found. The proof that this property really holds has to be done analytically after obtaining a promising candidate. 
\begin{figure}
    \centering
    \includegraphics[width=1.0\columnwidth]{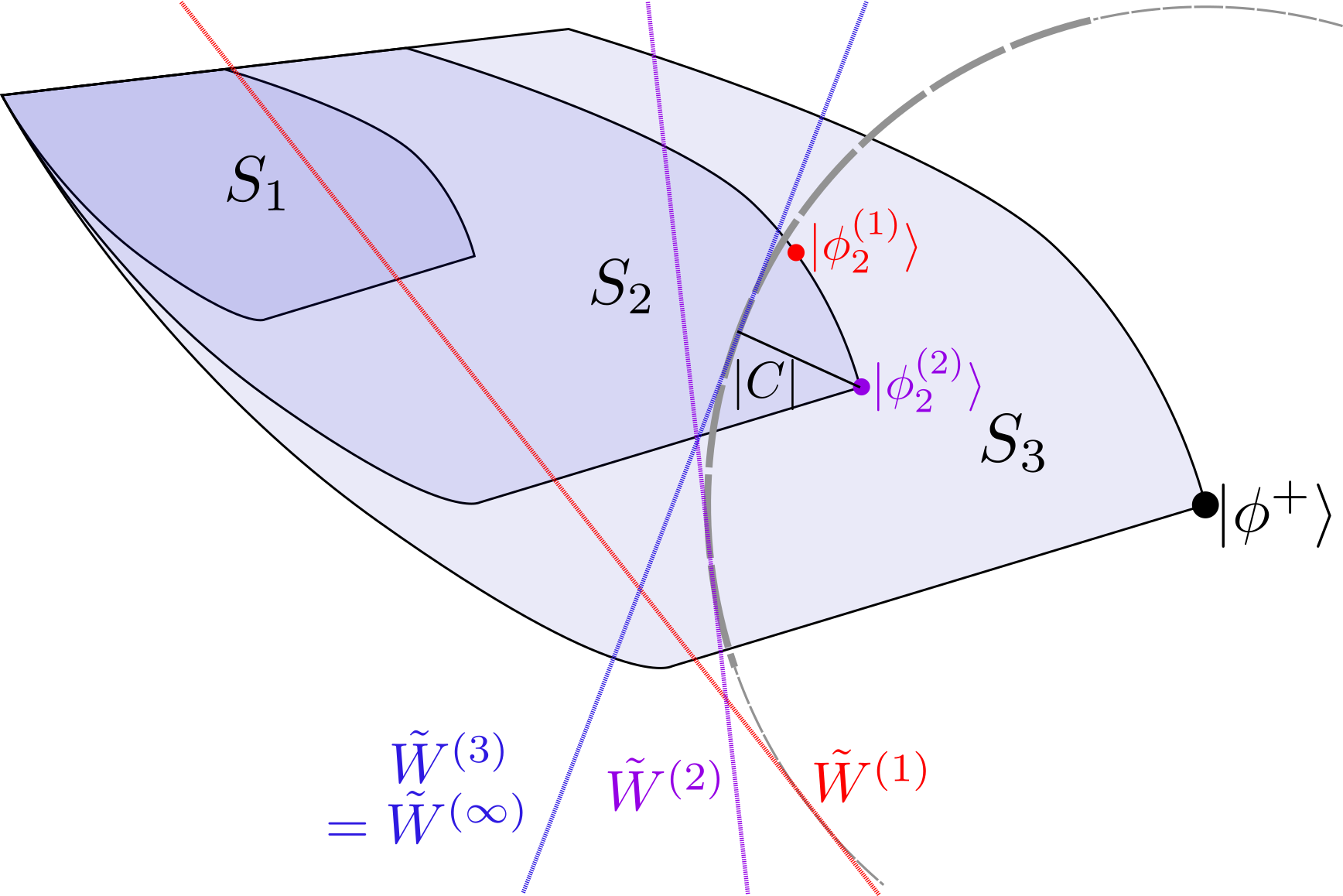}
    \caption{Illustration of the iterative algorithm described in the main text for $d=3$, $k=2$ and $\ket{\psi_0} = \ket{\phi^+}$. Initially, $\tilde{W}^{(1)}$ is found by the SDP in Eq.~(\ref{eq:sdpstep1}). Due to the constraint $\braket{\phi^+|W|\phi^+}=-1$, all $\tilde{W}^{(i)}$ must be tangent to the circle of radius 1 around the target state. For $\tilde{W}^{(1)}$, we then find the pure state $\ket{\phi_2^{(1)}}$ of Schmidt rank 2 that is furthest away from it, and add the constraint in Eq.~(\ref{eq:constraintC}) for it. After three iterations, no more such states can be found and $\tilde{W}^{(\infty)}$ is shifted to a proper witness candidate via Eq.~(\ref{eq:WtildetoWk}).}
    \label{fig:sets2}
\end{figure}
Finally, $\tilde{W}^{(\infty)}$  can be shifted to yield a proper Schmidt-number-$(k+1)$ witness via 
\begin{align}\label{eq:WtildetoWk}
    W_k = \frac1{1+ C}\tilde{W}^{(\infty)} - \frac{C}{1+C}\one_{d^2}.
\end{align}
This ensures the property $\trace(W_k\rho_k)\geq 0$ for all $\rho_k\in S_k$, as $\trace(W_k\rho_k) = \frac1{1+C}\trace(\tilde{W}^{(\infty)}\rho_k) - \frac{C}{1+C}\geq \frac{C}{1+C} - \frac{C}{1+C} \geq 0$. 

\section{Finding a witness that requires few measurement settings}\label{sec:setup}

We now apply the algorithm to find a witness that requires less measurements than the standard witness. For many experimental setups, it is reasonable to assume that the number of required measurement settings scales linearly with the number of matrix elements that are measured. Thus, we try different sets of $M$ of available coefficients in the witness, where the size of $M$ scales linearly in $d$. It turns out that the choice 
\begin{align}
    M &= \{(0,0,j,j)~|~0\leq j < d\} \cup \{(j,j,0,0)~|~0< j < d\} \cup \nonumber \\
    & \cup \{(0,j,0,j)~|~0< j < d\}\cup \{(j,0,j,0)~|~0< j < d\}\cup \nonumber \\
    & \cup\{(j,j,j,j)~|~0< j < d\}
\end{align}
yields good results while keeping the number of measurements in the order of $\mathcal{O}(d)$. We run the algorithm presented in the last section for $d=3,4,5$ and $6$ and different $k$, and obtain the following candidate for the witness, which will be shown to be a proper witness below:
\begin{align}\label{eq:shiftedWtilde}
    W_k = \frac1{1-\vert C\vert_k}\tilde{W} + \frac{\vert C\vert_k}{1-\vert C\vert_k}\one,
\end{align}
with 
\begin{align}\label{eq:Ck}
    \vert C\vert_k = \sqrt{\frac{d^2-4d+4k}{d^2}}
\end{align}
and 
\begin{align}\label{eq:Wtilde}
    \tilde{W} &= \left(1-\frac2d\right)\ketbra{00} + \delta_{d,3}\sum_{i=1}^{d-1}(\ketbraa{0i}{0i} + \ketbraa{i0}{i0}) - \nonumber \\
    &-\frac2d\sum_{i=1}^{d-1}(\ketbraa{00}{ii} + \ketbraa{ii}{00}) - \left(1-\frac2d\right)\sum_{i=1}^{d-1} \ketbra{ii}.
\end{align}
The fact that the algorithm yields the same $\tilde{W}$ for each choice of $k$, where only the thresholds $\vert C\vert_k$ vary, allows to run a single experiment to measure the overlap with $\tilde{W}$ and deduce a lower bound on the Schmidt number from that. 

The values $\vert C\vert_k$ correspond to the absolute value of the likely minimal overlap of $\tilde{W}$ with states $\ket{\phi_k}$ of Schmidt number $k$, which, according to the numerical minimization, are achieved by states of the form
\begin{align}\label{eq:specificfamilyk}
    \ket{\phi_k(\alpha)} = \alpha\ket{00} + \sqrt{\frac{1-\alpha^2}{k-1}} \sum_{i=1}^{k-1} \ket{ii}.
\end{align} Minimizing over this family of states yields $\min_{\alpha} \braket{\phi_k(\alpha)|\tilde{W}|\phi_k(\alpha)} = -\vert C\vert_k$. Note that this could be non-optimal and provides just an upper bound on the correct value of $C_k$. The last step is to show that this is indeed optimal, implying that the obtained candidate is actually a proper witness. To show this, we deduce that the lower bound from the dual SDP approximation using the sets $S_k^R$ coincides with the upper bound of $\vert C\vert_k$.
\begin{table}[t]
    \centering
    \begin{tabular}{r||r|r}
         $d=4$ & $\vert C\vert_k$ & $\vert C\vert_k^R$ \\
         \hline
         \hline
         $k = 1$ & 0.500 & 0.530 \\
        $k = 2$ & 0.707 & 0.715 \\
        $k = 3$ & 0.866 & 0.866 \\
        $k = 4$ & 1.000 & 1.000 \\
        \hline
        \rule{0pt}{4ex}
        $d=7$ & $\vert C\vert_k$ & $\vert C\vert_k^R$ \\
         \hline
         \hline
         $k = 1$ & 0.714 & 0.734 \\
        $k = 2$ & 0.769 & 0.772 \\
        $k = 3$ & 0.821 & 0.821 \\
        $k = 4$ & 0.869 & 0.869 \\
        $k = 5$ & 0.915 & 0.915 \\
        $k = 6$ & 0.958 & 0.958 \\
        $k = 7$ & 1.000 & 1.000 \\
    \end{tabular}
    \begin{tabular}{r||r|r}
         $d=11$ & $\vert C\vert_k$ & $\vert C\vert_k^R$ \\
         \hline
         \hline
         $k = 1$ & 0.818 & 0.825 \\
        $k = 2$ & 0.838 & 0.839 \\
        $k = 3$ & 0.858 & 0.858 \\
        $k = 4$ & 0.877 & 0.877 \\
        $k = 5$ & 0.895 & 0.895 \\
        $k = 6$ & 0.914 & 0.914 \\
        $k = 7$ & 0.932 & 0.932 \\
        $k = 8$ & 0.949 & 0.949 \\
        $k = 9$ & 0.966 & 0.966 \\
        $k = 10$ & 0.983 & 0.983 \\
        $k = 11$ & 1.000 & 1.000
    \end{tabular}
    \caption{The conjectured optimal threshold values $\vert C\vert_k$ from numerical minimization over Schmidt rank $k$ states, and the proven thresholds $\vert C\vert_k^R$ for $\tilde{W}$ in Eq.~(\ref{eq:Wtilde}) using the outer approximations $S_k^R$ and the dual of the SDP optimization: Whenever measuring an expectation value $\trace(\tilde{W}\rho) < -\vert C\vert_k^R$, the state $\rho$ has at least Schmidt number $k+1$. Note that for $k=1$, a different proof exists that establishes that a violation of $\vert C\vert_1$ instead of $\vert C\vert_1^R$ suffices to detect entanglement.}
    \label{tab:thresholds}
\end{table}
\begin{proposition}
For $k\neq 2$ and $d\geq4$, the witness in Eq.~(\ref{eq:shiftedWtilde}) is a Schmidt-number-$(k+1)$ witness.
\end{proposition}

\begin{proof}
We have to show that $\trace(\tilde{W}\rho_k) \geq C_k$ for all states $\rho_k\in S_k$, i.e., states with $\text{SN}(\rho_k) \leq k$.

Recall that a proper lower bound on
the minimal overlap is given by optimizing over the outer approximation $S_k^R$ of $S_k$: If
\opti{min}{\rho}{\trace(\tilde{W} \rho)}{\rho^\dagger = \rho, \trace(\rho) = 1, \rho \geq 0,;(\one \otimes R_{\frac1k})(\rho) \geq 0,}
where $R_p$ is given in Eq.~(\ref{eq:reductionmap}), is greater than or equal $C_k$, then this implies our claim. Thus, we effectively optimize over states in $S_k^R$. To solve this optimization, we convert this SDP into its dual form, given by Eq.~(\ref{eq:dualsdp}). Even though any feasible choice of $S$ provides a proper lower bound for $C_k$ such that $W_k$ is a witness, we aim to find the best choice, and numerically checking the optimal solutions for low values of $d$ seems to indicate that for $k\geq 2$, the optimal $S$ is given by $S=\ketbra{x(a,b)}$ with the unnormalized $\ket{x(a,b)} = a\ket{00} + b\sum_{i=1}^{d-1}\ket{ii}$ with $a,b\in \mathbb{R}$. As the dual SDP maximizies the smallest eigenvalue of  $\tilde{W} - (\one \otimes R_{\frac1k})(S)$, we first calculate the eigenvalues of this operator explicitly (up to multiplicities):
\begin{align}
    \lambda_1 &= -a^2,  \lambda_2 = -b^2, \lambda_3 = -b^2-\frac{d-2}{d},\\
    \lambda_{4,5} &= \frac{1}{2dk}[p(a,b) \pm \sqrt{p(a,b)^2 - 4q(a,b)}]),
\end{align}
with 
\begin{align}
    p(a,b) &= d[(1-k)a^2+(d-1-k)b^2], \\
    q(a,b) &= dk[d(k-d)a^2b^2 + (d-2)(k-1)a^2 +\nonumber \\
    &+(d-2)(d-1-k)b^2 + 4(d-1)ab-dk].
\end{align}
The minimal value of these $\lambda_i$ can only be one of the three eigenvalues $\lambda_1=-a^2, \lambda_3 = -b^2-\frac{d-2}d$ and $\lambda_5=\frac{1}{2dk}[p(a,b) - \sqrt{p(a,b)^2 - 4q(a,b)}]$.
For $k=1$ and $k=2$, the maximum minimal value of these is given when the three values coincide, which yields as the optimal value the root of an even, quartic polynomial that can be found efficiently, but yields a value larger than the one in Eq.~(\ref{eq:Ck}) (see Table~\ref{tab:thresholds}). For $k\geq3$, examining the numerical solutions for small $d$ suggests the optimal solution
\begin{align}
    a &= \frac{1}{\sqrt{k-1}}\sqrt{\vert C\vert_k + \frac{d-2}{d}},\\
    b &= \sqrt{\vert C\vert_k - \frac{d-2}{d}},
\end{align}
where $\vert C\vert_k$ is given in Eq.~(\ref{eq:Ck}).
This choice implies (after a lengthy, but elementary calculation) that $\lambda_3 = \lambda_5$, and for $d\geq3$ and $k\geq3$, it follows that
\begin{align}
    \lambda_1 - \lambda_3 &= \frac1{d(k-1)}\left(2-d+(k-2)\sqrt{d^2-4d+4k}\right) \nonumber \\
    & \geq\frac1{d(k-1)}(2-d+\sqrt{d^2-4d+4}) = 0,
\end{align}
thus $\lambda_3 \leq \lambda_1$, yielding the lower bound of $\lambda_3 = -b^2-\frac{d-2}{d} = -\vert C\vert_k$ to the optimization problem. As this lower bound coincides with the upper bound values from the minimization over the specific family of states in Eq.~(\ref{eq:specificfamilyk}), it must be optimal.

\begin{figure}
    \centering
    \includegraphics[width=0.8\columnwidth]{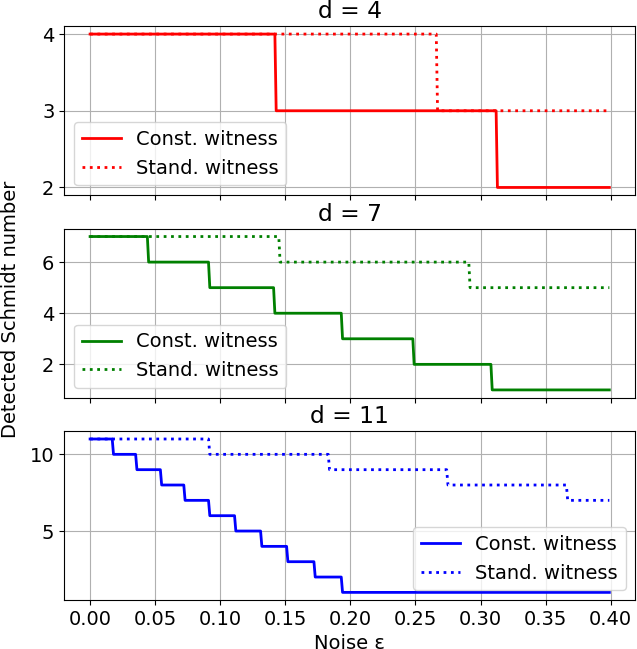}
    \caption{Detection power of the constructed witness in Eq.~(\ref{eq:Wtilde}) and the standard witness in Eq.~(\ref{eq:standardwitness}) for different dimensions as a function of white noise added to the state.}
    \label{fig:noise}
\end{figure}
While for $k=1$ the bound from the SDP is not optimal, it is possible to carry out the optimization over pure product states in this case using Lagrange multipliers, and one recovers the value $-\vert C\vert_1$. This leaves the choice of $k=2$ being the only case where we could not prove the numerical result to be a proper witness.
\end{proof}
While the proof does not work for $k=2$, it can still give bounds for these cases, meaning that a proper (probably non-optimal) witness can be built.

In Table~\ref{tab:thresholds}, we list the values of $\vert C\vert_k$ compared to $\vert C\vert_k^R$ for different values of $d$. It shows that for $k=1$ and $k=2$, the difference between the reduction map value and the conjectured $\vert C\vert_k$ is very small.

We stress that the precise number of measurement settings required to measure this witness depends on the specific experimental setup. In Appendix~\ref{app:tm}, we show that for a specific setup using photon temporal modes \cite{brecht2015photon}, the evaluation of the standard witness requires $2d^2-d$ measurement settings, whereas the newly constructed witness requires only $5d-4$ of them.

While our witness requires less measurement settings, it comes with the price of a reduced noise robustness. We model this by determining the maximum amount of white noise that can be added to the maximally entangled state before the witness fails to detect its $k$-dimensional entanglement. In particular, the noisy state is given by $\rho(\epsilon) = (1-\epsilon)\ketbra{\phi^+} + \epsilon \frac{\one_{d^2}}{d^2}$. The result is displayed in Fig.~\ref{fig:noise} for dimension $d=4$, $d=7$ and $d=11$. This indicates that our witness is
particularly useful in the low-noise regime. Finally, we note that we checked numerically that allowing for more measurement settings, i.e., increasing the size of the set $M$, increases the noise robustness and eventually yields the standard witness, as soon as the $d^2$ settings needed for its evaluation are included in $M$.  

\section{Conclusions}\label{sec:conclusions}

We presented an iterative algorithm that quickly generates candidates for Schmidt-number-$(k+1)$ witnesses which require only a number of measurement settings scaling linearly with the dimension. We applied the algorithm to find a witness candidate that requires knowledge of only $\mathcal{O}(d)$ of the matrix elements. To show that the numerical candidate is a witness, we employed a semi-definite program that gives lower bounds on the minimal overlap with states of certain Schmidt number, which in almost all cases turned out to be optimal, as it coincided with the upper bound obtained before.

While we applied the algorithm to a specific case, let us remark that it can be applied to many other experimental setups as well. The choice of setting certain coefficients of the witness to zero, for example, can be replaced by any set of linear or semi-definite constraints. With this, it could be fruitful to extend recent proposals for generating two-qudit entangled states. For instance, one could generate Schmidt number witnesses to certify entanglement dimensionality in the setup proposed in Ref.~\cite{puentes2020high}. 

\acknowledgments
We thank Laura Serino, Sophia Denker and Otfried Gühne for fruitful discussions. The authors acknowledge support by the QuantERA project QuICHE and the German Ministry of Education and Research (BMBF
Grant No. 16KIS1119K).

\appendix

\section{Number of measurements in photonic temporal mode setups}\label{app:tm}

As a concrete example, we consider the setup described in Ref.~\cite{brecht2015photon}. There, the produced density matrix is parameterized as
\begin{align}
    \rho = \sum_{ijkl=0}^{d-1} C_{ijkl} \ketbraa{ij}{kl},
\end{align}
and the measurements using quantum pulse gates are yielding the quantities
\begin{align}\label{eq:measurements}
    A_{c_A,c_B,\phi_A,\phi_B}^{m,n,p,q} & = c_A^2c_B^2C_{mpmp} +  s_A^2s_B^2C_{nqnq} + \nonumber \\
    &+ c_A^2s_B^2C_{mqmq} + s_A^2c_B^2C_{npnp} + \nonumber \\
    & + 2\operatorname{Re}[e^{i\phi_A}c_As_A(c_B^2C_{mpnp}+s_B^2C_{mqnq}) + \nonumber \\
    & \phantom{+2Re} + e^{i\phi_B}c_Bs_B(c_A^2C_{mpmq} + s_A^2C_{npnq}) + \nonumber \\
    & \phantom{+ 2\operatorname{Re}} + c_As_Ac_Bs_B (e^{i(\phi_A+\phi_B)}C_{mpnq}+ \nonumber \\
    & \phantom{2Re+c_as_ac_bs_Bcc}+e^{i(\phi_A-\phi_B)}C_{mqnp})]
\end{align}
with $c_A^2+s_A^2 = 1 = c_B^2+s_B^2$ and $0\leq m,n,p,q < d$. Choosing $c_{A,B} \in \{0, 1/\sqrt{2}, 1\}$ and $\phi_{A,B} \in \{0, \pi/2\}$ together with  $m,n,p,q$ allows to reconstruct all matrix elements. 

In order to measure the standard witness in Eq.~(\ref{eq:standardwitness}), the fidelity
\begin{equation}\label{eq:fidelityphip}
F = \braket{\phi^+|\rho|\phi^+} = \frac1d \sum_{i,j=0}^{d-1} C_{iijj}
\end{equation}
has to be measured. This requires obtaining $C_{iijj}$ for all $0\leq i,j < d$. The elements $C_{iiii}$ can be directly obtained from a measurement in Eq.~(\ref{eq:measurements}) with parameters $m=p=i$, $c_A = c_B = 1$ and the rest of the parameters arbitrary. The element $C_{iijj}$ with $i\neq j$ can be obtained by evaluating sums of eight measurements $A^{m,n,p,q}_{1/\sqrt{2},1/\sqrt{2},\phi_A,\phi_B}\equiv A^{m,n,p,q}_{\phi_A,\phi_B}$ as follows:
\begin{align} 
    2\operatorname{Re}[C_{iijj}] &= \phantom{-}A^{i,j,j,i}_{0,0} + A^{i,j,j,i}_{\pi,\pi} - A^{i,j,i,j}_{\pi/2,\pi/2} - A^{i,j,i,j}_{3\pi/2,3\pi/2}, \nonumber\\
    2\operatorname{Im}[C_{iijj}] &= -A^{i,j,j,i}_{\pi/2,0} - A^{i,j,j,i}_{3\pi/2,\pi} + A^{i,j,i,j}_{\pi,\pi/2} + A^{i,j,i,j}_{0,3\pi/2}.\label{eq:ReIm}
\end{align}
Note that the imaginary parts in Eq.~(\ref{eq:fidelityphip}) have to cancel, as $\rho$ is hermitian. Thus, we only have to measure the real parts of $C_{iijj}$. As $\operatorname{Re}(C_{iijj}) = \operatorname{Re}(C_{jjii})$, we have to measure the $d$ elements $C_{iiii}$, as well as the $d(d-1)/2$ elements $C_{iijj}$ for $i<j$, which need four measurement settings each.
Thus, in order to measure $F$, one needs $d + 2d(d-1) = 2d^2-d$ measurement settings.

For comparison, let us calculate the number of measurement settings needed to evaluate the derived witness in Eq.~(\ref{eq:Wtilde}) of the main text. Here, we focus on the case of $d>3$, meaning that we have to obtain the $d$ coefficients $C_{iiii}$ for $i=0,\ldots,d-1$, as well as the $d-1$ coefficients $C_{00ii}$ for $i=1,\ldots,d-1$. This requires a total of $d+4(d-1) = 5d-4$ measurement settings.

\bibliographystyle{apsrev4-1}
\bibliography{cite}

\end{document}